\documentclass{llncs}
\usepackage{tikz}
\usetikzlibrary{arrows,calc,shapes,decorations.pathreplacing}
\usetikzlibrary{decorations.pathmorphing}
\usetikzlibrary{decorations.markings}

\usepackage{amsmath,amsfonts,amssymb}
\usepackage{xspace}
\usepackage{url}

\usepackage{float}
\floatstyle{ruled}
\newfloat{myalgorithm}{thp}{lop}
\floatname{myalgorithm}{Algorithm}

\usepackage{algorithm}
\usepackage[noend]{algpseudocode}

\usepackage[numbers,sort&compress,longnamesfirst,sectionbib]{natbib}

\usepackage[font=small]{caption}

\usepackage{setspace,color}   % just for \NOTE
\newcounter{nummer}

\newcommand{\Wone}{\ensuremath{\textup{W}[1]}}

\newcommand{\Rd}{\ensuremath{\mathbb{R}^d}}

\newcommand{\Rpd}{\mathbb{R}_{\ge 0}^d}
\newcommand{\R}{\mathbb{R}}

\newcommand{\Oh}{\mathcal{O}}

\newcommand{\Vol}{\textsc{vol}}

\newcommand{\union}{\mathcal{U}}

\newcommand{\HYP}{\textsc{Hypervolume}}
\newcommand{\KMP}{\textsc{KMP}}
\newcommand{\CKMP}{\textsc{Cube-KMP}}
\newcommand{\UCKMP}{\textsc{Unitcube-KMP}}
\newcommand{\GRND}[1]{\textsc{\ensuremath{#1}-Grounded}}
\newcommand{\THYP}{\ensuremath{T_\HYP}}
\newcommand{\TKMP}{\ensuremath{T_\KMP}}
\newcommand{\TCKMP}{\ensuremath{T_\CKMP}}
\newcommand{\TUCKMP}{\ensuremath{T_\UCKMP}}
\newcommand{\TGRND}[1]{\ensuremath{T_\GRND{#1}}}

\def\polylog{\operatorname{polylog}}

\newcommand{\figref}[1]{Figure~\ref{fig:#1}}

\newcommand{\thmref}[1]{Theorem~\ref{thm:#1}}

\newcommand{\lemref}[1]{Lemma~\ref{lem:#1}}

\newcommand{\corref}[1]{Corollary~\ref{cor:#1}}
\newcommand{\secref}[1]{Section~\ref{sec:#1}}

\renewcommand{\epsilon}{\ensuremath{\varepsilon}}
\newcommand{\eps}{\ensuremath{\varepsilon}}

\newcommand{\ignore}[1]{}

\let\oldsqrt\sqrt
\def\hksqrt{\mathpalette\DHLhksqrt}
\def\DHLhksqrt#1#2{\setbox0=\hbox{$#1\oldsqrt{#2\,}$}\dimen0=\ht0
   \advance\dimen0-0.2\ht0
   \setbox2=\hbox{\vrule height\ht0 depth -\dimen0}%
   {\box0\lower0.4pt\box2}}
\renewcommand\sqrt\hksqrt
\renewcommand{\leq}{\leqslant}

\renewcommand{\le}{\leqslant}
\renewcommand{\ge}{\geqslant}

\catcode`@=11
\def\nphantom{\v@true\h@true\nph@nt}
\def\nvphantom{\v@true\h@false\nph@nt}
\def\nhphantom{\v@false\h@true\nph@nt}
\def\nph@nt{\ifmmode\def\next{\mathpalette\nmathph@nt}%
  \else\let\next\nmakeph@nt\fi\next}
\def\nmakeph@nt#1{\setbox\z@\hbox{#1}\nfinph@nt}
\def\nmathph@nt#1#2{\setbox\z@\hbox{$\m@th#1{#2}$}\nfinph@nt}
\def\nfinph@nt{\setbox\tw@\null
  \ifv@ \ht\tw@\ht\z@ \dp\tw@\dp\z@\fi
  \ifh@ \wd\tw@-\wd\z@\fi \box\tw@}
\newcount\minute \newcount\hour \newcount\hourMins
\def\now{\minute=\time \hour=\time \divide \hour by 60 \hourMins=\hour \multiply\hourMins by 60
  \advance\minute by -\hourMins \zeroPadTwo{\the\hour}:\zeroPadTwo{\the\minute}}

\def\today{\the\year-\zeroPadTwo{\the\month}-\zeroPadTwo{\the\day}}
\def\zeroPadTwo#1{\ifnum #1<10 0\fi #1}
\begin{document}

\title{\textbf{Bringing Order to Special Cases of  \\ Klee's Measure Problem}}

\author{
    Karl Bringmann%\footnote{Karl Bringmann is a recipient of the \emph{Google Europe Fellowship in Randomized Algorithms}, and this research is supported in part by this Google Fellowship.}
    %{\small {Max Planck Institute for Informatics}}  \\
    %{\small {Saarbr{\"u}cken, Germany}}
%           \affaddr{Max Planck Institute for Informatics}\\
%           \affaddr{Saarbr{\"u}cken, Germany}\\
%           \email{kbringma@mpi-inf.mpg.de}
}
\date{}
\institute{Max Planck Institute for Informatics}

\maketitle

\begin{abstract}
Klee's Measure Problem (KMP) asks for the volume of the union of $n$ axis-aligned boxes in $\R^d$. Omitting logarithmic factors, the best algorithm has runtime $\Oh^*(n^{d/2})$ [Overmars,Yap'91]. There are faster algorithms known for several special cases: \CKMP\ (where all boxes are cubes), \UCKMP\ (where all boxes are cubes of equal side length), \HYP\ (where all boxes share a vertex), and \GRND{k} (where the projection onto the first $k$ dimensions is a \HYP\ instance). 

In this paper we bring some order to these special cases by providing reductions among them. In addition to the trivial inclusions, we establish \HYP\ as the easiest of these special cases, and show that the runtimes of \UCKMP\ and \CKMP\ are polynomially related. 
More importantly, we show that any algorithm for one of the special cases with runtime $T(n,d)$ implies an algorithm for the general case with runtime $T(n,2d)$, yielding the first non-trivial relation between KMP and its special cases. This allows to transfer \Wone-hardness of \KMP\ to all special cases, proving that no $n^{o(d)}$ algorithm exists for any of the special cases assuming the Exponential Time Hypothesis. Furthermore, assuming that there is no \emph{improved} algorithm for the general case of KMP (no algorithm with runtime $\Oh(n^{d/2 - \eps})$) this reduction shows that there is no algorithm with runtime $\Oh(n^{\lfloor d/2 \rfloor /2 - \eps})$ for any of the special cases. Under the same assumption we show a tight lower bound for a recent algorithm for \GRND{2} [Y{\i}ld{\i}z,Suri'12].
\end{abstract}

%\fontsize{11pt}{13.5pt}\selectfont
%\fontsize{10pt}{12.4pt}\selectfont
%\newpage

\definecolor{light}{rgb}{0.85 0.85 0.85}
\definecolor{dark}{rgb}{0.65 0.65 0.65}

\section{Introduction}

\textbf{Klee's measure problem (\textbf{\KMP})} asks for the volume of the union of $n$ axis-aligned boxes in $\Rd$, where $d$ is considered to be a constant. 
This is a classic problem with a long history~\cite{Klee77,Bentley77,LeeuwenW81,OvermarsY91,Chan09,FredmanW78}. The fastest algorithm has runtime $\Oh(n^{d/2} \log n)$ for $d \ge 2$, given by Overmars and Yap~\cite{OvermarsY91}, which was slightly improved to $n^{d/2} 2^{\Oh(\log^*n)}$ by Chan~\cite{Chan09}. Thus, for over twenty years there has been no improvement over the runtime bound $n^{d/2}$. As already expressed in~\cite{Chan09}, one might conjecture that no \emph{improved} algorithm for \KMP\ exists, i.e., no algorithm with runtime $\Oh(n^{d/2 - \eps})$ for some $\eps > 0$. 

However, no matching lower bound is known, not even under reasonable complexity theoretic assumptions. The best unconditional lower bound is $\Omega(n \log n)$ for any dimension $d$~\cite{FredmanW78}. Chan~\cite{Chan09} proved that \KMP\ is \Wone-hard by giving a reduction to the $k$-Clique problem. Since his reduction has $k = d/2$, we can transfer runtime lower bounds from $k$-Clique to \KMP, implying that there is no $n^{o(d)}$ algorithm for \KMP\ assuming the Exponential Time Hypothesis (see~\cite{fptsurvey}). However, this does not determine the correct constant in the exponent. Moreover, Chan argues that since no ``purely combinatorial'' algorithm with runtime $\Oh(n^{k-\eps})$ is known for Clique, it might be that there is no such algorithm with runtime $\Oh(n^{d/2-\eps})$ for \KMP, but this does not rule out faster algorithms using, e.g., fast matrix multiplication techniques.

Since no further progress was made for \KMP\ for a long time, research turned to the study of special cases. Over the years, the following special cases have been investigated. For each one we list the asymptotically fastest results.

\begin{itemize}
\item \textbf{\CKMP:} Here the given boxes are cubes, not necessarily all with the same side length. This case can be solved in time $\Oh(n^{(d+2)/3})$ for $d \ge 2$~\cite{Bringmann10}. In dimension $d=3$ this has been improved to $\Oh(n \log^4 n)$ by Agarwal~\cite{Agarwal10}. In dimensions $d \le 2$ even the general case can be solved in time $\Oh(n \log n)$, the same bound clearly applies to this special case. As described in~\cite{Bringmann10}, there are simple reductions showing that the case of cubes is roughly the same as the case of ``$\alpha$-fat boxes'', where all side lengths of a box differ by at most a constant factor $\alpha$.

\item \textbf{\UCKMP:} Here the given boxes are cubes, all of the same side length. This is a specialization of \CKMP, so all algorithms from above apply. The combinatorial complexity of a union of unit cubes is $\Oh(n^{\lfloor d/2 \rfloor})$~\cite{Boissonnat98}. Using this, there are algorithms with runtime $\Oh(n^{\lfloor d/2 \rfloor} \polylog n)$~\cite{KaplanRSV07} and $\Oh(n^{\lceil d/2 \rceil - 1 + \frac{1}{\lceil d/2 \rceil}} \polylog n)$~\cite{Chan03}. Again, there is a generalization to ``$\alpha$-fat boxes of roughly equal size'', and any algorithm for \UCKMP\ can be adapted to an algorithm for this generalization~\cite{Bringmann10}.

\item \textbf{\HYP:} Here all boxes have a common vertex. Without loss of generality, we can assume that they share the vertex $(0,\ldots,0) \in \Rd$ and lie in the positive orthant~$\R_{\ge0}^d$. This special case is of particular interest for practice, as it is used as an indicator of the quality of a set of points in the field of Evolutionary Multi-Objective Optimization~\cite{ZitzlerT99,BeumeNE07,ZitzlerK04,Igel06}.
Improving upon the general case of \KMP, there is an algorithm with runtime $\Oh(n \log n)$ for $d=3$~\cite{BeumeFLIPV}.
The same paper also shows an unconditional lower bound of $\Omega(n \log n)$ for $d > 1$, while $\#P$-hardness in the number of dimensions was shown in~\cite{BF09b}. 
Recently, an algorithm with runtime $\Oh(n^{(d-1)/2} \log n)$ for $d \ge 3$ was presented in~\cite{YildizS12}.

\item \textbf{\GRND{k}:} Here the projection of the input boxes to the first $k$ dimensions is a \HYP\ instance, where $0 \le k \le d$, the other coordinates are arbitrary. This rather novel special case appeared in~\cite{YildizS12}, where an algorithm with runtime $\Oh(n^{(d-1)/2} \log^2 n)$ for $d \ge 3$ was given for \GRND{2}. 
\end{itemize}

Note that for none of these special cases \Wone-hardness is known, so there is no larger lower bound than $\Omega(n \log n)$ (for constant or slowly growing $d$), not even under reasonable complexity theoretic assumptions.
Also note that there are trivial inclusions of some of these special cases: Each special case can be seen as a subset of all instances of the general case. As such subsets, the following inclusions hold.
\begin{itemize}
\item $\UCKMP \subseteq \CKMP \subseteq \KMP$.
\item $\GRND{(k+1)} \subseteq \GRND{k}$ for all $k$.
\item $\GRND{d} = \HYP$ and $\GRND{0} = \KMP$.
\end{itemize}
This allows to transfer some results listed above to other special cases. E.g., the \CKMP\ algorithm with runtime $\Oh(n^{(d+2)/3})$ also applies to \UCKMP.

\subsection{Our results}

We present several reductions among the above four special cases and the general case of \KMP. They provide bounds on the runtimes needed for these variants and, thus, yield some order among the special cases.

Our first reduction relates \HYP\ and \UCKMP.

\begin{theorem} \label{thm:hypuckmp}
  If there is an algorithm for \UCKMP\ with runtime $\TUCKMP(n,d)$, then there is an algorithm for \HYP\ with runtime
  \begin{align*}
    \THYP(n,d) \le \Oh(\TUCKMP(n,d)).
  \end{align*}
\end{theorem}

Note that if \HYP\ were a subset of \UCKMP, then the same statement would hold, with the constant hidden by the $\Oh$-notation being~1. Hence, this reduction can nearly be seen as an inclusion.
Moreover, together with the trivial inclusions this reduction establishes \HYP\ as the easiest of all studied special cases. 

\begin{corollary} \label{cor:nlogn}
  For all studied special cases, \HYP, \UCKMP, \CKMP, and \GRND{k} (for any $0 \le k \le d$), we have the unconditional lower bound $\Omega(n \log n)$ for any $d > 1$.
\end{corollary}

One can find contradicting statements regarding the feasibility of a reduction as in \thmref{hypuckmp} in the literature.
On the one hand, existence of such a reduction has been mentioned in~\cite{YildizS12}. 
On the other hand, a newer paper~\cite{GuerreiroFE12} contains this sentence: ``Better bounds have been obtained for the KMP on unit
cubes ..., but reducing
the hypervolume indicator to such problems is not possible in general.''
In any case, to the best of our knowledge a proof of such a statement cannot be found anywhere in the literature.

Our second reduction substantiates the intuition that the special cases \CKMP\ and \UCKMP\ are very similar, by showing that their runtimes differ by at most a factor of $\Oh(n)$. Recall that $\UCKMP \subseteq \CKMP$ was one of the trivial inclusions. We prove an inequality in the other direction in the following theorem.

\begin{theorem} \label{thm:uckmpckmp}
  If there is an algorithm for \UCKMP\ with runtime $\TUCKMP(n,d)$, then there is an algorithm for \CKMP\ with runtime
  \begin{align*}
    \TCKMP(n,d) \le \Oh( n \cdot \TUCKMP(n,d) ).
  \end{align*}
\end{theorem}

Our third and last reduction finally allows to show lower bounds for all special cases. We show an inequality between the general case of \KMP\ and \GRND{2k}, in the opposite direction than the trivial inclusions. For this, we have to increase the dimension in which we consider \GRND{2k}.

\begin{theorem} \label{thm:thereduction}
  If there is an algorithm for \GRND{2k} in dimension $d+k$ with runtime $\TGRND{2k}(n,d+k)$, then there is an algorithm for \KMP\ in dimension $d$ with runtime
  \begin{align*}
    \TKMP(n,d) \le \Oh(\TGRND{2k}(n,d+k)).
  \end{align*}
\end{theorem}

Note that, if we set $k = d$, the special case \GRND{2k} in $d+k$ dimensions becomes \HYP\ in $2d$ dimensions. Since we established \HYP\ as the easiest variant, the above reduction allows to transfer \Wone-hardness from the general case to all special cases. Since the dimension is increased only by a constant factor, even the tight lower bound on the runtime can be transferred to all special cases.

\begin{corollary} \label{cor:Wone}
  There is no $n^{o(d)}$ algorithm for any of the special cases \HYP, \UCKMP, \CKMP, and \GRND{k}, assuming the Exponential Time Hypothesis.
\end{corollary}

We immediately get more precise lower bounds if we assume that no improved algorithm exists for \KMP\ (no algorithm with runtime $\Oh(n^{d/2 - \eps})$). 

\begin{corollary}
  If there is no improved algorithm for \KMP, then there is no algorithm with runtime 
  $\Oh(n^{\lfloor d/2 \rfloor /2 - \eps})$
  for any of \HYP, \UCKMP, \CKMP, and \GRND{k}, for any $\eps > 0$.
\end{corollary}

This shows the first lower bound for all studied special cases that is larger than $\Omega(n \log n)$. Note that there is, however, a wide gap to the best known upper bound of $\Oh(n^{(d+2)/3})$ for \HYP, \UCKMP, and \CKMP. 

Furthermore, setting $k=1$, \thmref{thereduction} immediately implies that the recent algorithm for \GRND{2} with runtime $\Oh(n^{(d-1)/2} \log^2 n)$~\cite{YildizS12} is optimal (apart from logarithmic factors and if there is no improved algorithm for \KMP).

\begin{corollary} \label{cor:grnd2}
  If there is no improved algorithm for \KMP, then
  there is no algorithm for \GRND{2} with runtime
  $\Oh(n^{(d-1)/2 - \eps})$ for any $\eps > 0$.
\end{corollary}

To simplify our runtime bounds, in some proofs we use the following technical lemma. Informally, it states that for any \GRND{k} algorithm with runtime $T(n,d)$ we have $T(\Oh(n),d) \le \Oh(T(n,d))$. Note that in this paper we hide by the $\Oh$-notation any functions depending solely on $d$.

\begin{lemma} \label{lem:technical}
  Fix $0 \le k \le d$ and $c > 1$. If there is an algorithm for \GRND{k} with runtime $\TGRND{k}(n,d)$ then there is another algorithm for \GRND{k} with runtime $\TGRND{k}'(n,d)$ satisfying
  \begin{align*}
    \TGRND{k}'(c n,d) \le \Oh(\TGRND{k}(n,d)).
  \end{align*}
\end{lemma}

\subsection{Notation and Organization}

A \emph{box} is a set of the form $B = [a_1,b_1]\times \ldots \times [a_d,b_d] \subset \Rd$, $a_i,b_i \in \R$, $a_i \le b_i$. 
A \emph{cube} is a box with all side lengths equal, i.e., $|b_1-a_1| = \ldots = |b_d - a_d|$. 
Moreover, a \emph{\KMP\ instance} is simply a set $M$ of $n$ boxes. In \CKMP\ all these boxes are cubes, and in \UCKMP\ all these boxes are cubes of common side length. In \HYP, all input boxes share the vertex $(0,\ldots,0) \in \Rd$, i.e., each input box is of the form $B = [0,b_1] \times \ldots \times [0,b_d]$. In \GRND{k}, the projection of each input box to the first $k$ dimensions is a \HYP\ box, meaning that each input box is of the form $B = [a_1,b_1] \times \ldots \times [a_d,b_d]$ with $a_1 = \ldots = a_k = 0$.

We write the usual Lebesgue measure of a set $A \subseteq \Rd$ as $\Vol(A)$. For sets $R,A \subseteq \Rd$ we write $\Vol_R(A) := \Vol(R \cap A)$, the volume of $A$ restricted to $R$. 
For a \KMP\ instance $M$ we let $\union(M) := \bigcup_{B \in M} B$. To shorten notation we write $\Vol(M) := \Vol(\union(M))$ and $\Vol_R(M) := \Vol(R \cap \union(M))$. 

In the next section we present the proof of \thmref{hypuckmp}. In \secref{uckmpckmp} we prove \thmref{uckmpckmp}. The proof of \thmref{thereduction} is split into \secref{thereductionOne} and \secref{thereductionTwo}: We first give the reduction for \GRND{2} and then generalize this result to $\GRND{2k}$, $k > 1$. We prove the technical \lemref{technical} in \secref{technical} and close with an extensive list of open problems in \secref{conclusions}.

\section{\HYP\ $\leq$ \UCKMP}
\label{sec:hypuckmp}

In this section we prove \thmref{hypuckmp} by giving a reduction from \HYP\ to \UCKMP. 

Given an instance of \HYP, let $\Delta$ be the largest coordinate of any box. We extend all boxes to cubes of side length $\Delta$, yielding a \UCKMP\ instance. In this process, we make sure that the new parts of each box will not lie in the positive orthant $\Rpd$, but in the other orthants, as depicted in \figref{hypuckmp}. This means that the volume of the newly constructed cubes - restricted to $\Rpd$ - is the same as the volume of the input boxes. To compute this restricted volume, we compute the volume of the constructed \UCKMP\ instance once with and once without an additional cube $C = [0,\Delta]^d$. From this we can infer the volume of the input \HYP\ instance. 

\vspace{-0.2cm}

\begin{figure*}
\begin{center}

\begin{footnotesize}
\begin{tikzpicture}[scale=1.6]

\definecolor{cube}{rgb}{0.85 0.85 0.85};

\path (0.1,1.) coordinate (A);
\path (0.3,0.7) coordinate (B);
\path (0.4,0.5) coordinate (C);
\path (0.75,0.3) coordinate (D);

\fill[cube] (0.,0.) rectangle (A);
\fill[cube] (0.,0.) rectangle (B);
\fill[cube] (0.,0.) rectangle (C);
\fill[cube] (0.,0.) rectangle (D);

\draw (0.,0.) rectangle (A);
\draw (0.,0.) rectangle (B);
\draw (0.,0.) rectangle (C);
\draw (0.,0.) rectangle (D);

\draw[->,thick] (-0.1,0.) -- (1.1,0.);
\draw[->,thick] (0.,-0.1) -- (0.,1.1);

\draw[->,decorate, decoration={snake,amplitude=1.5pt}] (1.25,0.5) -- (1.9,0.5);

\begin{scope}[shift={(3.,0.)}]

\path (0.1,1.) coordinate (A);
\path (0.3,0.7) coordinate (B);
\path (0.4,0.5) coordinate (C);
\path (0.75,0.3) coordinate (D);
\path (1.,1.) coordinate (Om);

\fill[cube] (0.,0.) rectangle (A);
\fill[cube] (0.,0.) rectangle (B);
\fill[cube] (0.,0.) rectangle (C);
\fill[cube] (0.,0.) rectangle (D);

\path (A) ++(-1.,-1.) coordinate (AA);
\path (B) ++(-1.,-1.) coordinate (BB);
\path (C) ++(-1.,-1.) coordinate (CC);
\path (D) ++(-1.,-1.) coordinate (DD);

\draw (AA) rectangle (A);
\draw (BB) rectangle (B);
\draw (CC) rectangle (C);
\draw (DD) rectangle (D);

\draw[dashed,thick] (0,0) rectangle (1,1);

\path (1.,1.) node[below left] {$C$};
\draw[decorate,decoration={brace,raise=6pt,amplitude=4pt}]
    (1,1)--(1,0);
\path (1.2,0.5) node [right] {$\Delta$};

\draw[->,thick] (-0.1,0.) -- (1.1,0.);
\draw[->,thick] (0.,-0.1) -- (0.,1.1);

\end{scope}

\end{tikzpicture}
\end{footnotesize}
\end{center}

\caption{\label{fig:hypuckmp}
Construction in the proof of \thmref{hypuckmp}.
}
\end{figure*}

%\vspace{-0.9cm}

\section{\UCKMP\ $\ge$ \CKMP}
\label{sec:uckmpckmp}

In this section we prove \thmref{uckmpckmp} by giving a reduction from \CKMP\ to \UCKMP. 

Given a \CKMP\ instance, let $C$ be the cube with smallest side length. We will compute the \emph{contribution} $v$ of $C$, i.e., the volume of space that is contained in~$C$ but no other cube. Having this, we can delete $C$ and recurse on the remaining boxes. Adding up yields the total volume of the input instance.

To compute $v$, we modify each cube such that it becomes a cube of $C$'s side length and its restriction to $C$ stays the same, as depicted in \figref{uckmpckmp}. Applying this construction to all input boxes, we get a \UCKMP\ instance that, inside $C$, looks the same as the input \CKMP\ instance. Computing the volume of this new instance once with and once without $C$ allows to infer $v$.

%Given a \CKMP\ instance we will show how to compute its volume assuming that we can solve \UCKMP. For this, consider the smallest cube $C$ in the input instance, i.e., the one with smallest side length. We will compute the volume $v$ of space that is contained in~$C$ but no other box of the input instance, see \figref{uckmpckmp}. Then we recursively compute the volume $w$ of the union of the $n-1$ boxes that are left after deleting~$C$. The sum of both, $v+w$, is the volume of the input instance. 

%To compute $v$, we may modify each cube $B \ne C$ in any manner such that the restriction of $B$ to $C$ stays the same. One can easily see that we can modify all other cubes such that their restriction to $C$ remains untouched and they become cubes of $C$'s side length, as depicted in \figref{uckmpckmp}. Applying this construction to all input boxes, we get a \UCKMP\ instance that, inside $C$, looks the same as the input \CKMP\ instance. Now we compute the volume of this new instance, once with and once without the cube $C$. Subtracting both values, we get exactly $v$. 

\vspace{-0.1cm}

\begin{figure*}
\begin{center}

\begin{footnotesize}
\begin{tikzpicture}[scale=1.1]

\definecolor{cube}{rgb}{0.85 0.85 0.85};

\path (0.1,0.4) coordinate (A);
\path (0.3,0.7) coordinate (B);
\path (0.5,0.15) coordinate (C);

\path (0.72,1.2) coordinate (D);
\path (-0.2,1.1) coordinate (E);

\path (A) ++(-1.3,-1.3) coordinate (AA);
\path (B) ++(-1.05,1.05) coordinate (BB);
\path (C) ++(1.1,-1.1) coordinate (CC);

\path (D) ++(1.7,-1.7) coordinate (DD);
\path (E) ++(-1.25,-1.25) coordinate (EE);

\fill[cube] (0,0) rectangle (1,1);

\fill[white] (A) rectangle (AA);
\fill[white] (B) rectangle (BB);
\fill[white] (C) rectangle (CC);
\fill[white] (D) rectangle (DD);
\fill[white] (E) rectangle (EE);

\draw[thick] (0,0) rectangle (1,1);

\draw (A) rectangle (AA);
\draw (B) rectangle (BB);
\draw (C) rectangle (CC);
\draw (D) rectangle (DD);
\draw (E) rectangle (EE);

\path (0.45,0.5) node {$v$};
\path (1.1,1) node[below left] {$C$};

\draw[->,decorate, decoration={snake,amplitude=1.5pt}] (2.7,0.5) -- (3.6,0.5);

\begin{scope}[shift={(4.8,0.)}]

\path (0.1,0.4) coordinate (A);
\path (0.3,0.7) coordinate (B);
\path (0.5,0.15) coordinate (C);

\path (0.72,1.) coordinate (D);

\path (A) ++(-1,-1) coordinate (AA);
\path (B) ++(-1,1) coordinate (BB);
\path (C) ++(1,-1) coordinate (CC);

\path (D) ++(1,-1) coordinate (DD);

\fill[cube] (0,0) rectangle (1,1);

\fill[white] (A) rectangle (AA);
\fill[white] (B) rectangle (BB);
\fill[white] (C) rectangle (CC);
\fill[white] (D) rectangle (DD);

\draw[thick,dashed] (0,0) rectangle (1,1);

\draw (A) rectangle (AA);
\draw (B) rectangle (BB);
\draw (C) rectangle (CC);
\draw (D) rectangle (DD);

\path (0.45,0.5) node {$v$};
\path (1.1,1) node[below left] {$C$};

\end{scope}

\end{tikzpicture}
\end{footnotesize}
\end{center}

\caption{\label{fig:uckmpckmp}
Construction in the proof of \thmref{uckmpckmp}.
}
\end{figure*}

\vspace{-0.9cm}

\section{2-Grounded $\ge$ \KMP}
\label{sec:thereductionOne}

We first show the reduction of \thmref{thereduction} for \GRND{2}, i.e., we show $\TKMP(n,d) \le \Oh(\TGRND{2}(n,d+1))$ by giving a reduction from \KMP\ to \GRND{2}. 
This already implies \corref{grnd2} and lays the foundations for the complete reduction given in the next section. 

We begin by showing the reduction for $d=1$. As a second step we show how to generalize this to larger dimensions.

\subsection{Dimension $d=1$} \label{sec:grndTwoDOne}

We want to give a reduction from \KMP\ in 1 dimension to \GRND{2} in 2 dimensions. Note that the latter is the same as \HYP\ in 2 dimensions.
Let $M$ be an instance of \KMP\ in 1 dimension, i.e., a set of $n$ intervals in $\R$. 
We will reduce the computation of $\Vol(M)$ to two instances of \GRND{2}.

Denote by $x_1 < \ldots < x_{m}$ the endpoints of all intervals in $M$ (if all endpoints are distinct then $m=2n$). We can assume that $x_1 = 0$ after translation. Consider the boxes 
\begin{align*}
  A_i := [m-i-1,m-i] \times [x_i,x_{i+1}]
\end{align*}
in $\R^2$ for $1 \le i \le m-1$, as depicted in \figref{grndkmp}. Denote the union of these boxes by~$A$. Note that the volume of box $A_i$ is the same as the length of the interval $[x_i,x_{i+1}]$. This means that we took the chain of intervals $\{[x_i,x_{i+1}]\}$ and made it into a staircase of boxes $\{A_i\}$, where each box has the same volume as the corresponding interval.

\vspace{-0.3cm}

\begin{figure*}
\begin{center}

\begin{footnotesize}
\begin{tikzpicture}[scale=3.]

\draw (0,0) -- (1,0);

\foreach \i / \x in {1/0.,2/0.15,3/0.35,4/0.6,5/0.72,6/1.} {
  \draw (\x,0.03) -- (\x,-0.03) node[below] {$x_{\i}$};
}

\draw[thick] (0.17,0.04) -- (0.15,0.04) -- (0.15,-0.04) -- (0.17,-0.04);
\draw[thick] (0.58,0.04) -- (0.6,0.04) -- (0.6,-0.04) -- (0.58,-0.04);
\draw[thick] (0.15,0) -- (0.6,0);
\path[thick] (0.15,0.03) -- (0.6,0.03) node[midway,above] {$I$};

\draw[->,decorate, decoration={snake,amplitude=1.5pt}] (1.2,0.) -- (1.5,0.);

\begin{scope}[shift={(1.9,-0.5)}]

\foreach \x / \y / \j / \k in {0./0.15/0.8/1., 0.15/0.35/0.6/0.8, 0.35/0.6/0.4/0.6, 0.6/0.72/0.2/0.4, 0.72/1./0./0.2} {
  \fill[dark] (\j,\x) rectangle (\k,\y);
}

\foreach \x / \y / \j / \k in {0./0.15/0.8/1., 0.15/0.35/0.6/0.8, 0.35/0.6/0.4/0.6, 0.6/0.72/0.2/0.4, 0.72/1./0./0.2} {
  \fill[light] (0,0) rectangle (\k,\x);
}

\draw[->] (0,-0.02) -- (0,1.05);
\draw[->] (-0.02,0) -- (1.05,0);

\foreach \i / \x / \y / \j in {2/0.15/0.2/1, 3/0.35/0.4/2, 4/0.6/0.6/3, 5/0.72/0.8/4, 6/1./1./5} {
  \draw (0.03,\x) -- (-0.03,\x) node[left] {$x_{\i}$};
}
\draw (0.03,0) -- (-0.03,0) node[left] {$0=x_1$};

\foreach \i / \x / \y / \j in {1/0./0./0, 2/0.15/0.2/1, 3/0.35/0.4/2, 4/0.6/0.6/3, 5/0.72/0.8/4, 6/1./1./5} {
  \draw (\y,0.03) -- (\y,-0.03) node[below] {$\j$};
}

\foreach \x / \y / \j / \k in {0./0.15/0.8/1., 0.15/0.35/0.6/0.8, 0.35/0.6/0.4/0.6, 0.6/0.72/0.2/0.4, 0.72/1./0./0.2} {
  \draw (\j,\x) rectangle (\k,\y);
}

\draw[thick] (0,0) rectangle (0.8,0.6) node[above right] {$C_I$};

\path (0.9,0.07) node {$A_1$};
\path (0.7,0.25) node {$A_2$};
\path (0.5,0.47) node {$A_3$};
\path (0.3,0.66) node {$A_4$};
\path (0.1,0.85) node {$A_5$};

\end{scope}

\end{tikzpicture}
\end{footnotesize}
\end{center}

\caption{\label{fig:grndkmp}
The left hand side depicts all endpoints $0=x_1 \le \ldots \le x_6$ of a 1-dimensional \KMP\ instance. An input interval $I$ is indicated. The right hand side shows the result of our transformation. Each interval $[x_i,x_{i+1}]$ to the left corresponds to a box $A_i$ to the right. The interval $I$ gets mapped to the box $C_I$. The shaded regions depict the set $A$ ({\protect\tikz \protect\fill[dark] (0,0) rectangle (1.5ex,1.6ex);}, the union of all $A_i$) and the set $T_0$ ({\protect\tikz \protect\fill[light] (0,0) rectangle (1.5ex,1.6ex);}). 
}
\end{figure*}
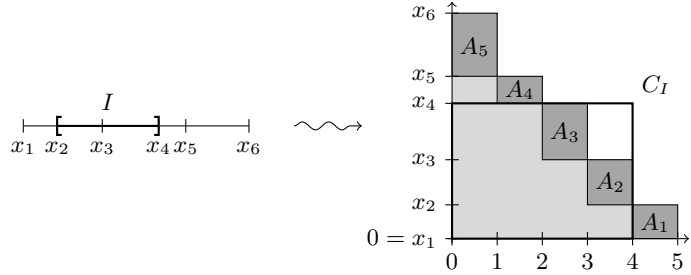

Now consider an interval $I = [x_j,x_k] \in M$. We construct the box 
\begin{align*}
  C_I := [0,m-j] \times [0,x_k],
\end{align*}
also shown in \figref{grndkmp}. Then $C_I$ contains the boxes $A_i$ with $j \le i < k$ and (its interior) has no common intersection with any other box $A_i$. This is easily seen as $A_i \subseteq C_I$ iff $m-i \le m-j$ and $x_{i+1} \le x_k$. Hence, for any interval $I \in M$ we constructed a box $C_I$ that contains exactly those boxes $A_i$ whose corresponding interval $[x_i,x_{i+1}]$ is contained in $I$, or in other words
\begin{align*}
  [x_i,x_{i+1}] \subseteq I &\Leftrightarrow A_i \subseteq C_I,  \\
  \Vol([x_i,x_{i+1}] \cap I) &= \Vol(A_i \cap C_I).  
\end{align*}

From these properties it follows that the volume of $C_I$ restricted to $A$ is the same as the length of~$I$, i.e.,
\begin{align*}
  \Vol_A(C_I) = \Vol(I).
\end{align*}
Furthermore, considering the whole set $M$ of intervals, the interval $[x_i,x_{i+1}]$ is contained in some interval in~$M$ iff the box $A_i$ is contained in some box in $C_M := \{C_I \mid I \in M \}$. This yields
\begin{align*}
  \Vol(M) = \Vol_A(C_M).
\end{align*}

It remains to reduce the computation of $\Vol_A(C_M)$ to two \GRND{2} instances. For this we consider 
\begin{align*}
  T_0 := \bigcup_{1 \le i \le m} C_{[x_i,x_i]}.
\end{align*}
Informally speaking, $T_0$ consists of all points ``below'' $A$, as depicted in \figref{grndkmp}. Note that no set $A_j$ is contained in $T_0$. Moreover, we consider the set $T_1 := T_0 \cup A$. Observe that we can write
\begin{align*}
  T_1 = \bigcup_{1 \le i \le m-1} C_{[x_i,x_{i+1}]},
\end{align*}
since $A_i \subseteq C_{[x_i,x_{i+1}]}$. 
Note that both sets $T_0$ and $T_1$ are unions of $\Oh(n)$ \GRND{2} boxes. Informally, $T_0$ is the maximum \GRND{2} instance that has $\Vol_A(T_0) = 0$, and $T_1$ is the minimum \GRND{2} instance with $\Vol_A(T_1) = \Vol(A)$. 
Now, we can compute $\Vol_A(C_M)$ as follows.
\begin{lemma}
  In the above situation we have
  \begin{align*}
    \Vol_A(C_M) = \Vol(A) + \Vol(T_0 \cup \union(C_M)) - \Vol(T_1 \cup \union(C_M)).
  \end{align*}
\end{lemma}
\begin{proof}
  Set $U := \union(C_M)$.
  Using $T_0 \subseteq T_1$ and $A = T_1 \setminus T_0$ in a sequence of simple transformations, we get
  \begin{align*}
    \Vol(T_1 \cup U) - \Vol(T_0 \cup U)
    &= \Vol((T_1 \cup U) \setminus (T_0 \cup U))  \\
    &= \Vol((T_1 \setminus T_0) \setminus U)  \\
    &= \Vol(A \setminus U)  \\
    &= \Vol(A) - \Vol(A \cap U)  \\
    &= \Vol(A) - \Vol_A(C_M),
  \end{align*}
  which proves the claim. \qed
\end{proof}
Note that $\Vol(A) = \sum_i \Vol(A_i) = \sum_i |x_{i+1}-x_i| = |x_m - x_1|$ is trivial.
Also note that both sets $T_0$ and $T_1$ are the union of $\Oh(n)$ \GRND{2} boxes, so that $\Vol(T_b \cup \union(C_M))$ can be seen as a \GRND{2} instance of size $\Oh(n)$, for both $b \in \{0,1\}$. Hence, we reduced the computation of the input instance's volume $\Vol(M)$ to $\Vol_A(C_M)$ and further to the \GRND{2} instances $\Vol(T_0 \cup \union(C_M))$ and $\Vol(T_1 \cup \union(C_M))$.

As we have to sort the given intervals first, we get
\begin{align*}
  \TKMP(n,1) \le \Oh(\TGRND{2}(\Oh(n),2) + n \log n).
\end{align*}
Note that this inequality alone gives no new information, as already Klee~\cite{Klee77} showed that $\TKMP(n,1) \le \Oh(n \log n)$. However, we get interesting results when we generalize this reduction to higher dimensions in the next section.

\subsection{Larger Dimensions}

In this section we show how the reduction from the last section carries over to larger dimensions, yielding a reduction from \KMP\ in $d$ dimensions to \GRND{2} in $d+1$ dimensions. This implies $\TKMP(n,d) \le \Oh(\TGRND{2}(n,d+1))$.  

Assume we are given a \KMP\ instance $M$ in dimension $d$. The idea is that we use the dimension doubling reduction from the last section on the first dimension and leave all other dimensions untouched. More precisely, for a box $B \in M$ let $\pi_1(B)$ be its projection onto the first dimension and let $\pi_*(B)$ be its projection onto the last $d-1$ dimensions, so that $B = \pi_1(B) \times \pi_*(B)$. Now follow the reduction from the last section on the instance $M' := \{\pi_1(B) \mid B \in M\}$. This yields sets $A$, $T_0$, $T_1$, and a box $C_I$ for each $I \in M'$.

We set $C_B := C_{\pi_1(B)} \times \pi_*(B)$ and $C_M = \{C_B \mid B \in M\}$. A possible way of generalizing $A$ would be to set $A'' := A \times \R^{d-1}$. Then we would be interested in $\Vol_{A''}(C_M)$, which can be seen to be exactly $\Vol(M)$. This definition of $A''$ is, however, not simple enough, as it is not a difference of \GRND{2} instances (unlike $A = T_1 \setminus T_0$). To give a different definition, assume (after translation) that all coordinates of the input instance are non-negative and let $\Delta$ be the maximal coordinate in any dimension. We set $A' := A \times [0,\Delta]^{d-1}$ and still get the same volume $\Vol_{A'}(C_M) = \Vol(M)$. This allows to generalize $T_0$ and $T_1$ to $T_0' := T_0 \times [0,\Delta]^{d-1}$ and $T_1' := T_1 \times [0,\Delta]^{d-1}$, while still having
\begin{align*}
  \Vol_{A'}(C_M) = \Vol(A') + \Vol(T_0' \cup \union(C_M)) - \Vol(T_1' \cup \union(C_M)).
\end{align*}
Note that $T_0'$ and $T_1'$ are also a union of $\Oh(n)$ \GRND{2} boxes, so a volume such as $\Vol(T_0' \cup \union(C_M))$ can be seen as a \GRND{2} instance.
This completes the reduction and yields the time bound
\begin{align*}
  \TKMP(n,d) \le \Oh(\TGRND{2}(\Oh(n),d+1) + n \log n).
\end{align*}
Using the lower bound $\Omega(n \log n)$ of \corref{nlogn} we can hide the additional $n \log n$ in the first summand. 
Moreover, first using the technical \lemref{technical} we can finally simplify this to the statement of \corref{grnd2},
\begin{align*}
  \TKMP(n,d) \le \Oh(\TGRND{2}(n,d+1)).
\end{align*}

\section{2k-Grounded $\ge$ \KMP}
\label{sec:thereductionTwo}

In this section we prove \thmref{thereduction} in its full generality, building on the redcutions from the last section. Recall that we want to give a reduction from \KMP\ in dimension $d$ to \GRND{2k} in dimension $d+k$. 

Let $M$ be a \KMP\ instance. After translation we can assume that in every dimension the minimal coordinate among all boxes in $M$ is 0. Denoting the largest coordinate of any box by $\Delta$ we thus have $B \subseteq [0,\Delta]^d = \Omega^d$ for all $B \in M$, where $\Omega := [0,\Delta]$.

The first steps of generalizing the reduction from the last section to the case $k > 1$ are straight forward. We want to use the dimension doubling reduction from \secref{grndTwoDOne} on each one of the first $k$ dimensions. For any box $B \in \R^d$ denote its projection onto the $i$-th dimension by $\pi_i(B)$, $1 \le i \le k$, and its projection onto dimensions $k+1,\ldots,d$ by $\pi_*(B)$. We use the reduction from \secref{grndTwoDOne} on each dimension $1 \le i \le k$, i.e., on each instance $M^{(i)} := \{\pi_i(B) \mid B \in M\}$, yielding sets $A^{(i)}, T_0^{(i)}, T_1^{(i)}$, and a box $C_I^{(i)}$ for each $I \in M^{(i)}$.
We assume that all these sets are contained in $[0,\Delta]^2 = \Omega^2$, meaning that all coordinates are upper bounded by $\Delta$ (this holds after possibly increasing the $\Delta$ we chose before).

For a box $B \in M$ we now define 
\begin{align*}
  C_B := C^{(1)}_{\pi_1(B)} \times \ldots \times C^{(k)}_{\pi_k(B)} \times \pi_*(B).
\end{align*}
This is a box in $\R^{d+k}$, it is even a \GRND{2k} box, as its projection onto the first $2k$ coordinates has the vertex $(0,\ldots,0)$. Set $C_M := \{C_B \mid B \in M\}$. 
As shown by the following lemma, we want to determine the volume $\Vol_{A}(C_M)$ in
\begin{align*}
  A := A^{(1)}\times \ldots \times A^{(k)} \times [0,\Delta]^{d-k}.
\end{align*}

\begin{lemma} \label{lem:VolACM}
  In the above situation we have
  \begin{align*}
    \Vol(M) = \Vol_A(C_M).
  \end{align*}
\end{lemma}
\begin{proof}
  Denote by $x_1^{(i)} < \ldots < x_{m_i}^{(i)}$ the coordinates of all boxes in $M$ in the $i$-th dimension. We can express $\Vol(M)$ in terms of the boxes 
  \begin{align*}
    E_{j_1,\ldots,j_d} := [x_{j_1}^{(1)}, x_{j_1+1}^{(1)}] \times \ldots \times [x_{j_d}^{(d)}, x_{j_d+1}^{(d)}],
  \end{align*}
  for $1 \le j_i < m_i$. Since each such box is either completely included in some box in $M$ or does not contribute to $\Vol(M)$, we have 
  \begin{align*}
    \Vol(M) = \sum_{j_1,\ldots,j_d} [ E_{j_1,\ldots,j_d} \subseteq \union(M) ] \cdot \Vol(E_{j_1,\ldots,j_d}),
  \end{align*}
  where $[X]$ is 1 if $X$ is true, and 0 otherwise.
  
  Recall from the reduction in \secref{grndTwoDOne} that $A^{(i)} = A_1^{(i)} \cup \ldots \cup A^{(i)}_{m_i - 1}$, and there is a one-to-one correspondence between intervals $[x_j^{(i)},x_{j+1}^{(i)}]$ and 2-dimensional boxes $A_j^{(i)}$, in particular both have the same volume. This carries over to a one-to-one correspondence between $d$-dimensional boxes $E_{j_1,\ldots,j_d}$ and $(d+k)$-dimensional boxes 
  \begin{align*}
    E'_{j_1,\ldots,j_d} := A_{j_1}^{(1)} \times \ldots \times A_{j_k}^{(k)} \times [x_{j_{k+1}}^{(k+1)}, x_{j_{k+1}+1}^{(k+1)}] \times \ldots \times [x_{j_d}^{(d)}, x_{j_d+1}^{(d)}],
  \end{align*}
  in particular both have the same volume.
  
  Additionally, recall that an interval $I \in M^{(i)}$ includes $[x_j^{(i)},x_{j+1}^{(i)}]$ if and only if the 2-dimensional box $C_I^{(i)}$ contains $A_j^{(i)}$. If $I$ does not include $[x_j^{(i)},x_{j+1}^{(i)}]$, then $\Vol(I \cap [x_j^{(i)},x_{j+1}^{(i)}]) = 0$, and we also have $\Vol(C_I^{(i)} \cap A_j^{(i)}) = 0$.
  Hence, we have for any $B \in M$ that $E_{j_1,\ldots,j_d} \subseteq B$ if and only if $E'_{j_1,\ldots,j_d} \subseteq C_B$, which implies that $E_{j_1,\ldots,j_d} \subseteq \union(M)$ if and only if $E'_{j_1,\ldots,j_d} \subseteq \union(C_M)$. Furthermore, $E'_{j_1,\ldots,j_d}$ is either included in $\union(C_M)$ or does not contribute to $\Vol(C_M)$.
  In total, we get
  \begin{align*}
    \Vol(M) &= \sum_{j_1,\ldots,j_d} [ E_{j_1,\ldots,j_d} \subseteq \union(M) ] \cdot \Vol(E_{j_1,\ldots,j_d})  \\
    &= \sum_{j_1,\ldots,j_d} [ E'_{j_1,\ldots,j_d} \subseteq \union(C_M) ] \cdot \Vol(E'_{j_1,\ldots,j_d}) = \Vol_A(C_M),
  \end{align*}
  which completes the proof. \qed
\end{proof}

Unfortunately, the set $A$ is not simply a difference of two \GRND{2k} instances. 
Thus, the hard part is to reduce the computation of $\Vol_{A}(C_M)$ to \GRND{2k} instances, which we will do in the remainder of this section.
For $1 \le i \le k$ and $b \in \{0,1\}$ we set
\begin{align*}
  \tilde{T}_b^{(i)} := \Omega^{2(i-1)} \times T_b^{(i)} \times \Omega^{d+k-2i}.
\end{align*}
This set in $\Omega^{d+k}$ consists of all points $x$ whose projection to dimensions $2i-1$ and $2i$ is contained in $T_b^{(i)}$.
Note that each set $\tilde{T}_b^{(i)}$ can be written as the union of $\Oh(n)$ \GRND{2k} boxes, since $T_b^{(i)}$ is the union of $\ell = \Oh(n)$ \GRND{2} boxes in $\R^2$, i.e., $T_b^{(i)} = \bigcup_{j=1}^{\ell} C_j$, so that we may write 
$\tilde{T}_b^{(i)} = \bigcup_{j=1}^{\ell} \Omega^{2(i-1)} \times C_j \times \Omega^{d-k-2i}$. Thus, we can use an algorithm for \GRND{2k} to compute any volume of the form $\Vol(\tilde{T}_b^{(i)} \cup V)$, where~$V$ is a union of $\Oh(n)$ \GRND{2k} boxes. 
%However, we cannot compute, e.g., the volume of $T_0^{(1)} \times \ldots \times T_0^{(k)}$, since this set can be written as the union of \GRND{2k} boxes, but the number of boxes in such a representation is too large. Thus, we have to be careful to only reduce to volumes as in the first example.

Furthermore, define for $S \subseteq [k]$
\begin{align*}
  D_S := \bigg( \bigcup_{i \in S} \tilde{T}_1^{(i)} \bigg) \cup \bigcup_{i \in [k] \setminus S} \tilde{T}_0^{(i)}.
\end{align*}
Note that $D_S \subseteq D_{S'}$ holds for $S \subseteq S'$. We can express $A$ using the sets $D_S$ as shown by the following lemma.

\begin{lemma} \label{lem:AcapDS}
  In the above situation we have
  \begin{align*}
    A = \bigcap_{1 \le i \le k} D_{\{i\}} \setminus D_\emptyset.
  \end{align*}
\end{lemma}
\begin{proof}
  We have $D_\emptyset = \bigcup_{i \in [k]} \tilde{T}_0^{(i)}$ and $\tilde{T}_0^{(i)} \subseteq \tilde{T}_1^{(i)}$ for all $i$, implying
  \begin{align*}
    D_{\{i\}} = \tilde{T}_1^{(i)} \cup D_\emptyset.
  \end{align*}
  This yields
  \begin{align*}
    \bigcap_{1 \le i \le k} D_{\{i\}} \setminus D_\emptyset
    &= \bigg( \bigcap_{1 \le i \le k} \tilde{T}_1^{(i)} \bigg) \setminus \bigcup_{i \in [k]} \tilde{T}_0^{(i)}.
  \end{align*}
  A point $x = (x_1,\ldots,x_{d+k}) \in \R^{d+k}$ is in the set on the right hand side if and only if it has the following three properties:
  \begin{itemize}
    \item $x_i \in [0,\Delta]$ for all $1\le i \le d+k$,
    \item $(x_{2i-1},x_{2i})$ lies in $T_1^{(i)}$ for all $1 \le i \le k$,
    \item $(x_{2i-1},x_{2i})$ does not lie in $T_0^{(i)}$ for all $1 \le i \le k$.
  \end{itemize}
  Since $A^{(i)} = T_1^{(i)} \setminus T_0^{(i)}$, this description captures exactly $A^{(1)}\times \ldots \times A^{(k)} \times [0,\Delta]^{d-k} = A$, finishing the proof. \qed
\end{proof}

Moreover, each $D_S$ can be written as the union of $\Oh(n)$ \GRND{2k} instances, since the same was true for the sets $\tilde{T}_b^{(i)}$. Hence, we can use an algorithm for \GRND{2k} to compute the volume
\begin{align*}
  H_S := \Vol(D_S \cup \union(C_M)).
\end{align*}
Next we show that we can compute $\Vol_A(C_M)$ from the $H_S$ by an interesting usage of the inclusion-exclusion principle.

\begin{lemma} \label{lem:last}
  In the above situation we have
  \begin{align*}
    \Vol_A(C_M) = \Vol(A) + \sum_{S \subseteq [k]} (-1)^{|S|} H_S.
  \end{align*}
\end{lemma}
\begin{proof}
  In this proof we write for short $U := \union(C_M)$. We clearly have
  \begin{align} \label{eq:firststep}
    \Vol(A) - \Vol_A(U) = \Vol(A \setminus U).
  \end{align}
  We first show 
  \begin{align} \label{eq:toshow}
    \Vol(A \setminus U) = \sum_{\emptyset \ne S \subseteq [k]} (-1)^{|S|+1} (H_S - H_\emptyset),
  \end{align}
  and simplify the right hand side later.
  Using $A = \bigcap_{1 \le i \le k} D_{\{i\}} \setminus D_\emptyset$ (\lemref{AcapDS}) and the inclusion-exclusion principle we arrive at
  \begin{align*}
    \Vol(A \setminus U) 
    &= \Vol\Big( \bigcap_{1 \le i \le k} D_{\{i\}} \setminus (D_\emptyset \cup U) \Big)  \\
    &= \sum_{\emptyset \ne S \subseteq [k]} (-1)^{|S|+1} \Vol\bigg( \bigcup_{i \in S} D_{\{i\}} \setminus (D_\emptyset \cup U) \bigg).
  \end{align*}
  Note that $D_S = \bigcup_{i \in S} D_{\{i\}}$, so that the above equation simplifies to
  \begin{align} \label{eq:undnocheine}
    \Vol(A \setminus U) 
    &= \sum_{\emptyset \ne S \subseteq [k]} (-1)^{|S|+1} \Vol\bigg( D_S \setminus (D_\emptyset \cup U) \bigg)
  \end{align}    
  
  Using the definition of $H_S$ and $D_\emptyset \subseteq D_S$ for any $S \subseteq [k]$, we get
  \begin{align*}
    H_S - H_\emptyset 
    &= \Vol(D_S \cup U) - \Vol(D_\emptyset \cup U) \\
    &= \Vol((D_S \cup U) \setminus (D_\emptyset \cup U))  \\
    &= \Vol( D_S \setminus (D_\emptyset \cup U) ).
  \end{align*}
  Plugging this into (\ref{eq:undnocheine}) yields (\ref{eq:toshow}).
  
  Observe that 
  \begin{align*}
    \sum_{\emptyset \ne S \subseteq [k]} (-1)^{|S|+1} (- H_\emptyset)
    &= - H_\emptyset.
  \end{align*}    
  This allows to further simplify (\ref{eq:toshow}) to
  \begin{align*}
    \Vol(A \setminus U) = \sum_{S \subseteq [k]} (-1)^{|S|+1} H_S.
  \end{align*}
  Plugging this into (\ref{eq:firststep}) yields the desired equation. \qed
\end{proof}

As $\Vol(A) = \Delta^{d-k} \cdot \prod_{1 \le i \le k} \Vol(A^{(i)})$ is trivial, we have reduced the computation of $\Vol_A(C_M)$ to $2^k = \Oh(1)$ instances of \GRND{2k}, each consisting of $\Oh(n)$ boxes.
During the construction of these instances we need to sort the coordinates, so that we need additional time $\Oh(n \log n)$. This yields
\begin{align*}
  \TKMP(n,d) \le \Oh(\TGRND{2k}(\Oh(n),d+k) + n \log n).
\end{align*}
Because of the lower bound from \corref{nlogn}, we have $\TGRND{2k}(\Oh(n),d+k) = \Omega(n \log n)$, so we can hide the second summand in the first,
\begin{align*}
  \TKMP(n,d) \le \Oh(\TGRND{2k}(\Oh(n),d+k)).
\end{align*}
We may use the technical \lemref{technical} to get rid of the inner $\Oh$:
This lemma guarantees an algorithm with runtime $\TGRND{2k}'(n,d)$ such that we get
\begin{align*}
  \TKMP(n,d) \le \Oh(\TGRND{2k}'(\Oh(n),d+k)) \le \Oh(\TGRND{2k}(n,d+k)).
\end{align*}
This finishes the proof.

\section{Proof of \lemref{technical}} \label{sec:technical}

In this section we prove the technical \lemref{technical}.

\begingroup
\def\thelemma{\ref{lem:technical}}
\begin{lemma}
  Fix $0 \le k \le d$ and $c > 1$. If there is an algorithm for \GRND{k} with runtime $\TGRND{k}(n,d)$ then there is another algorithm for \GRND{k} with runtime $\TGRND{k}'(n,d)$ satisfying
  \begin{align*}
    \TGRND{k}'(c n,d) \le \Oh(\TGRND{k}(n,d)).
  \end{align*}
\end{lemma}
\endgroup

\begin{proof}
  Let $M$ be an instance of \GRND{k} of size $|M| = n$. Denote by $z_1\le \ldots \le z_{2n}$ the coordinates in the first dimension of all boxes in $M$. Let $a_1 := z_{(1-\alpha)n}$ and $b_1 := z_{(1+\alpha)n}$, where $\alpha := 1/(3d)$.
  Denote by $\{x_1 \le a_1\}$ the set of all points $x=(x_1,\ldots,x_d) \in \Rd$ with $x_1 \le a_1$, similarly for $\{a_1 \le x_1 \le b_1\}$ and $\{x_1 \ge b_1\}$. We consider the three \KMP\ instances
  \begin{align*}
    M_a &:= \{ B \cap \{x_1 \le a_1\} \mid B \in M\},  \\
    M_b &:= \{ B \cap \{x_1 \ge b_1\} \mid B \in M\},  \\
    M' &:= \{ B \cap \{a_1 \le x_1 \le b_1\} \mid B \in M\}.
  \end{align*}
  Note that all three instances can be seen as \GRND{k} instances: Projected onto the first $k$ dimensions, all boxes in $M_a$ share the vertex $(0,\ldots,0)$, all boxes in $M_b$ share the vertex $(b_1,0,\ldots,0)$, and all boxes in $M'$ share the vertex $(a_1,0,\ldots,0)$, so after translation they all share the vertex $(0,\ldots,0)$. Moreover, $\{x_1 < a_1\}$ and $\{x_1 > b_1\}$ contain only $(1-\alpha)n$ coordinates of boxes in $M$ in the first dimension. Hence, there are at most $(1-\alpha)n$ boxes intersecting $\{x_1 < a_1\}$, so after deleting boxes with volume 0 we get $|M_a| \le (1-\alpha)n$, similarly for $M_b$. This reasoning does not work for $M'$, it might even be that all $n$ boxes are present in $M'$: If a box has left coordinate smaller than $a_1$ and right coordinate larger than $b_1$, then none of its coordinates is seen in $\{a_1 \le x_1 \le b_1\}$, although it has non-empty intersection with $\{a_1 \le x_1 \le b_1\}$. However, such a box in $M'$ is trivial in the first dimension: its coordinates in the first dimension are simply $[a_1,b_1]$. If all boxes in $M'$ were trivial in the first dimension, then $M'$ would clearly be simpler than the input instance. Although this is not the case, we can bound the number of boxes in $M'$ that are non-trivial in the first dimension: Since there are at most $2 \alpha n$ coordinates~$z_i$ in $[a_1,b_1]$, all but at most $2 \alpha n$ boxes in $M'$ are trivial in the first dimension. Thus, also $M'$ is easier than the input instance $M$, in a certain sense.
  
  Note that $\Vol(M) = \Vol(M_a) + \Vol(M_b) + \Vol(M')$ and $M_a,M_b$ are strictly easier than $M$, as they contain at most $(1-\alpha)n$ boxes. We have to simplify $M'$ further. 
  For this, we use the same construction as above (on $M$ and dimension 1) on $M'$ and dimension 2, i.e., we split by coordinates in dimension 2 at $a_2$ and $b_2$. This yields three \GRND{k} instances. Two of them contain at most $(1-\alpha)n$ boxes. The third one, $M''$, may contain up to $n$ boxes. However, all but at most $4 \alpha n$ of these boxes are trivial in the first and second dimension, meaning that their projection onto the first 2 dimensions is $[a_1,b_1]\times [a_2,b_2]$. 
  
  Iterating this reduction $d$ times yields $2d$ instances of \GRND{k} containing at most $(1-\alpha)n$ points and one instance $M^*$ that may contain up to $n$ boxes. However, all but at most $2d \, \alpha n$ of these boxes are trivial in all $d$ dimensions, meaning that they are equal to $[a_1,b_1] \times \ldots \times [a_d,b_d]$. Since all boxes in $M^*$ are contained in $[a_1,b_1] \times \ldots \times [a_d,b_d]$, if any such trivial box exists, the volume of $M^*$ is trivial. Otherwise $M^*$ only contains at most $2d \, \alpha n = \frac23 n \le (1-1/(3d))n = (1-\alpha)n$ boxes. Thus, we have reduced the computation of $\Vol(M)$ to $2d+1$ instances of \GRND{k} with at most $(1-\alpha)n$ boxes each. The reduction itself can be made to run in $\Oh(n)$ time. Hence, if we solve the reduced problems by an algorithm with runtime $\TGRND{k}(n,d)$, then we get an algorithm with runtime $\TGRND{k}'(n,d)$ satisfying
  \begin{align*}
    \TGRND{k}'(n,d) \le (2d+1) \TGRND{k}((1-\alpha)n,d) + \Oh(n).
  \end{align*}
  As every algorithm for \GRND{k} at least has to read its whole input, we can hide the $\Oh(n)$ by the first term,
  \begin{align*}
    \TGRND{k}'(n,d) \le \Oh( \TGRND{k}((1-\alpha)n,d) ),
  \end{align*}
  or, 
  \begin{align*}
    \TGRND{k}'(n/(1-\alpha),d) \le \Oh( \TGRND{k}(n,d) ).
  \end{align*}
  Repeating this construction an appropriate number of times we can increase the constant $1/(1-\alpha)$ to any constant $c > 1$, while the factor on the right hand side is still bounded by a constant. This finally yields an algorithm with runtime satisfying
  \begin{align*}
    \TGRND{k}''(cn,d) \le \Oh(\TGRND{k}(n,d)).
  \end{align*} \qed
\end{proof}

\section{Conclusion}  \label{sec:conclusions}

We presented reductions between the special cases \CKMP, \UCKMP, \HYP, and \GRND{k} of Klee's measure problem. These reductions imply statements about the runtime needed for these problem variants. We established \HYP\ as the easiest among all studied special cases, and showed that the variants \CKMP\ and \UCKMP\ have polynomially related runtimes.
Moreover, we presented a reduction from the general case of \KMP\ to \GRND{2k}. This allows to transfer \Wone-hardness from \KMP\ to all special cases, proving that no $n^{o(d)}$ algorithm exists for any of the special cases assuming the Exponential Time Hypothesis. Moreover, assuming that no improved algorithm exists for \KMP, we get a tight lower bound for a recent algorithm for \GRND{2}, and a lower bound of roughly $n^{(d-1)/4}$ for all other special cases. Thus, we established some order among the special cases of Klee's measure problem.

Our results lead to a number of open problems, both asking for new upper and lower bounds:

\begin{itemize}
\item Is there a polynomial relation between $\HYP$ and $\UCKMP$, similar to  $\CKMP$ and $\UCKMP$, or do both problems have significantly different runtimes?
\item Show that no improved algorithm exists for \KMP, e.g., assuming the Strong Exponential Time Hypothesis, as has been done for the Dominating Set problem, see~\cite{fptsurvey}. Or give an improved algorithm.
\item Assuming that no improved algorithm for \KMP\ exists, we know that the optimal runtimes of \HYP\ and \CKMP/\UCKMP\ are of the form $n^{c_d\cdot d \pm \Oh(1)}$, with $c_d \in [1/4,1/3]$. Determine the correct value of~$c_d$.
\item Generalize the $\Oh(n^{(d-1)/2} \log^2 n)$ algorithm for \GRND{2}~\cite{YildizS12} to an $\Oh(n^{(d-k)/2 + o(1)})$ algorithm for \GRND{2k}. This would again be optimal by \thmref{thereduction}.
\item We showed the relation $\TKMP(n,d) \le \Oh(\TGRND{2k}(n,d+k))$. Show an inequality in the opposite direction, i.e., a statement of the form $\TGRND{k}(n,d) \le \Oh(\TKMP(n,d'))$ with $d' < d$.
\end{itemize}

%\newpage

\renewcommand{\bibsection}{\section*{References}}

%\bibliographystyle{myenglishabbrvnat}
%\bibliography{hypervol}

\end{document}